\documentclass[sn-mathphys-num]{sn-jnl}
\usepackage{amsmath,amssymb,amsthm}
\usepackage{braket}
\usepackage{graphicx}
\usepackage{url}
\usepackage{bbm}
\newcommand{\bp}{{\mathbf p}}
\newcommand{\bx}{{\mathbf x}}
\newcommand{\n}{{\mathbf w}}
\newcommand{\be}{\begin{equation}}
\newcommand{\ee}{\end{equation}}
\newcommand{\BR}{\mathbb{R}}
\newcommand{\BC}{\mathbb{C}}
\newcommand{\bpi}{\boldsymbol{\pi}}
\newcommand{\bv}{{\mathbf v}}
\DeclareMathOperator\Sec{Sec}

\newcommand{\cL}{{\mathcal{L}^\uparrow}}
\newcommand{\tD}{\tilde{D_{(m)}}}
\newcommand{\om}{|\bp|}
\newcommand{\f}{{\mathbf f}}
\newcommand{\tH}{{\tilde{\mathcal{H}}}}
\newcommand{\tU}{\tilde{U}}
\renewcommand{\H}{{\mathcal{H}}}
\def\bM{{\bf M}}

\newcommand{\mP}{\mathcal{P}}
\newcommand{\m}{{\mathbf a}}
\newcommand{\Br}{{\mathbf r}}
\newtheorem{theorem}{Theorem}
\newtheorem{Remark}{Remark}
\newtheorem{Proposition}[theorem]{Proposition}
\begin{document}
\title{The explicit form of the unitary representation of the Poincar\'e group for vector-valued wave functions (massive and massless), with applications to photon's localization and position operators}
\author{Arkadiusz Jadczyk\email{azjadczyk@gmail.com}}
\affil{Ronin Institute, Montclair, NJ 07043}

\abstract{

We geometrically derive the explicit form of the unitary representation of the Poincar\'e group for vector-valued wave functions and use it to apply speed-of-light boosts to simple polarization basis to end up with Hawton-Baylis photon position operator with commuting components. We give explicit formulas for other photon boost eigenmodes. We investigate the underlying affine connections on the light cone in momentum space and find that while Pryce connection is metric semi-symmetric, the flat Hawton-Baylis connection is not semi-symmetric. Finally we discuss localizability of photon states on closed loops and show that photon states on the circle, both unnormalized improper states and finite norm wave packet smeared over washer-like regions are strictly localized not only with respect to Hawton-Baylis operators with commuting components but also with respect to the noncommutative Jauch-Piron-Amrein POV measure.}

\maketitle


\section{Introduction}
In classical relativistic field theory the electric and magnetic field strengths are combined into a closed two-form $F$ over the Minkowski spacetime. Since $F$ is closed, $dF=0,$ it admits a four-vector potential $A,$ (a 1-form) so that $F=dA.$ While in classical theory $A$ is convenient, but not strictly necessary, when we move to quantum theory $A$ is almost unavoidable. There, when describing the interaction of the electromagnetic field with charged matter (particles or fields), $A$ has a simple geometrical interpretation as a connection one-form in a complex line bundle, and $F$ becomes a curvature of this connection (physicists sometimes like to use the term ``a non-integrable phase factor" for the fact that we are dealing with a connection form of a no-zero curvature). This simple and beautiful geometrical picture gets, however, somewhat lost when we move to the quantum theory of electromagnetic field itself (both first and second quantized).\footnote{A promising new way of keeping the geometry alive also in quantum field theory (within the algebraic framework) has been suggested by D. Buchholz et al. - c.f.  \cite{buch} and references therein.} Photons - the quanta of the electromagnetic field, are treated there as relativistic elementary particles, and as such are described by irreducible unitary representations of the Poincar\'e group. The whole machinery of Lie algebras, Casimir operators, Lie groups, induced representations and ``little groups" is being brought forward, and geometry gets almost completely forgotten. Physicists construct one-particle Hilbert spaces and multi-particle Fock spaces (while the first quantization is a `miracle', the second quantization is a functor), move into the algebra of operators, and this is a whole new world, with little place for differential geometry.
V.S. Varadarajan in his classic monograph ``Geometry of Quantum Theory'', chapter ``Representations in vector bundles and wave equations'' comes very close to fulfilling this task, unfortunately, when it comes to photons, the chapter ends with the sentence ``We do not get into these ideas here." - \cite[p. 371]{var}.

The group-theoretical analysis of elementary relativistic quantum systems lead to the concept of imprimitivity systems, developed by G.W. Mackey (cf. e.g. \cite[Ch. VI]{var} and references therein), and to the associated concept of the localization of elementary quantum particles. A.S. Wightman \cite{wightman} applied these concepts to the study of localizability of quantum mechanical systems and came to conclusion confirming the previous analysis of T.D. Newton and E.P Wigner \cite{nw}, namely that photons (as well as other particles of rest mass zero and helicity $\geq 1$) are covariantly non-localizable in a strict sense of an imprimitivity system based on the $3$-d Euclidean group acting on $\BR^3.$.

J.M. Jauch and C. Piron \cite{jp}, developed a concept of ``weak localizability'' replacing projection-valued measure by POV (positive operator-valued) measures, and A.O. Amrein \cite{amrein} proved that there exist
photon states strictly POV-localized in arbitrarily small regions of space, while, more recently, I. and Z. Bialynicki-Birula \cite{bb3} argued that photons cannot be sharply localized because of a kind of complementarity between magnetic and electric energy localization.

Closely related to the problem of a photon's localization is the problem of existence of the photon position operator ${\bf Q}.$ The problem is not exactly the same since a given vector-valued operator may have infinitely many representations in terms of POV measures (except when $Q_i$'s commute, and then there is a distinguished projection-valued spectral measure). It is known \cite{mourad} that the standard requirements of the covariance with respect to the Euclidean group and inversions lead to a unique ${\bf Q},$ - known as the Pryce photon position operator\footnote{The proof of uniqueness provided in this reference has a hole, as it requires an additional restriction on the form of the operator. But the hole in the proof can be completed, and no extra assumptions are in fact necessary.}, the trouble is that the components $Q_i$ do not commute, which makes the simple probabilistic interpretation for the photon's localization problem impossible.\footnote{While $Q_i$ admits a natural decomposition with respect to a POV measure, K. Kraus \cite{kraus} has shown that there exists more than one such measure, so that the question appears which one of them is more natural than others, and why?}  B.S. Skagerstam \cite{skagerstam} interpreted noncommutativity of components of the Pryce operators $Q_i^{\mbox{PR}}$ in terms of the curvature of a connection in a photon's momentum space\footnote{Skagerstam is using Jackiw's concept of ``three-cocycle'' \cite{jackiw}, developed before under a different name (generalized imprimitivity system) by the present author \cite{jadczyk}.}, and applied it to derivation of the Berry phase for a photon \cite{ciao}. The same idea has been discussed before by I. and Z. Bialynicki-Birula \cite{bb2}, except that in their paper the same connection has been derived independently of the photon's position operator question.

Abandoning the requirement of a covariance with respect to the $3D$ rotation group opens the way towards the construction of a huge family of ``position operators" for photons, including those with commuting components and therefore admitting a unique spectral decomposition. M. Hawton \cite{hawton} has started a whole series of works in this direction. Commuting of the components of the position operator implies a flat (curvature zero) connection on the positive light cone in momentum space. Connections of this type (no curvature, but torsion) has been long ago investigated by A. Staruszkiewicz \cite{staru1,staru2}, who required that the parallel transport preserves the natural (degenerate) metric, the volume form, and that the connection is semi-symmetric (a constraint on the torsion). Hawton and Baylis \cite{bh} investigated a particular photon position operator ${\bf Q}^{\mbox{HB}}$ with commuting components, the Hawton-Baylis operator , with axial symmetry. \footnote{If Berry phase is related to a non-zero curvature of the connection responsible for the parallel transport, as it is usually assumed, then any such position operator must lead to a vanishing Berry phase.} Recently these ideas have been further developed by Dobrski et al. \cite{dob1,dob2}.

The present work started with the realization that there is an apparent discrepancy between the geometrical picture of the photon wave function as a section of the tangent or cotangent bundle\footnote{Whenever a (non-degenerate) metric is available, there is no need to distinguishes between tangent and cotangent bundles. The distinction between the two becomes relevant only in premetric formulations of electrodynamics - c.f. \cite{hehl} and references therein.}, and the way the Lorentz boosts act on vector-valued functions within the unitary representation of the Poincar\'e group considered in all papers dealing with the photon wave function and with photon position operators.\footnote{For the photon wave function see e.g. \cite{bb0}, while for photon position operators discussed in a similar context as that taken in the present paper see \cite{dob2} and references therein.} In the present paper we start with a geometrical description of massive vector fields transforming naturally under the Poincar\'e group. We work in the momentum space\footnote{A good introduction to photon's wave mechanics, in both momentum and position space, as well as to the photon's localization problem, can be found in the review article \cite{keller1}, and in the monograph \cite{keller2}, by Ole Keller.} and discuss the natural unitary representation acting on the Hilbert space of sections of the tangent bundle of the positive mass hyperboloid, square integrable with respect to the natural Lorentz invariant measure. The natural Riemannian metric appearing there becomes degenerate in the limit $m=0.$ For $m>0,$ taking the positive square root of this metric we split the tangent space at each point into two mutually orthogonal parts:  the longitudinal-transversal split with respect to the mass-independent standard Euclidean metric. This enables us to twist the unitary representation along the longitudinal part. We obtain an explicit form of so obtained irreducible unitary representation and then take the limit $m=0$, which now is finite (though not irreducible). Only then the Lie algebra of the Poincar\'e group acquires the standard form.

We then use the explicit form of the boost unitary operators to obtain rather unexpected result: the polarization basis used for constructing the teleparallel connection by Hawton and Baylis \cite{bh} can be obtained by taking the speed-of-light limit in the $z$-direction of simple Hertz-type potentials ${\bf e}_1\sim \bp\times\n$ and ${\bf e}_2\sim \bp\times(\bp\times \n), \n=(0,0,1).$ We then discuss the Hawton-Baylis connection and the associated photon position operator with commuting components, and compare it to the Pryce connection with non-zero curvature and torsion. In particular we find that the Pryce connection is metric semi-symmetric, while Hawton-Baylis connection does not have the semi-symmetry property.

In the last part of this paper we analyze photon states $\f_\ell$ localized on loops $\ell$ in photon's position space. They are given as simple superpositions of plane waves localized at the points of the loop. For the particular case when the loop is a circle on a $z=0$ plane we show that $\f_l$ are localized on the circle not only with respect to the Jauch-Piron-Amrein POV measure $F(\Delta)$  (what was known before), but also (which came as a surprise), with respect to ${\bf Q}^{\mbox{HB}}.$

We construct simple wave packets made of these circle states that provide normalized photon states localized in arbitrarily small washer-like regions of space, again both with respect to  ${\bf Q}^{\mbox{HB}}$ and $F(\Delta).$

\noindent {\bf Notation}\\
We work in the momentum representation, spacetime signature $(-+++).$. Coordinates $(p^0,p^1,p^2,p^3)=(p^0,\bp),$ $p^2=\bp^2-(p^0)^2.$ We will write $\BR^3$ to denote the $3$-dimensional real vector space of the momentum vectors $\bp.$ Greek indices $\mu,\nu,\rho,\sigma$ run from $0$ to $3$. Latin indices $i,j,k,l$ run from $1$ to $3$. We use the Greek letter $\alpha,$ also running from $1$ to $3$ for numbering the basis vectors in momentum space. Summation over repeated indices is implied.
Only positive energies, $p^0>0,$ are being used. We will use the notation $\bpi$ to denote the dimensionless unit vector in the direction of the momentum $\bp$: $\bpi=\bp/|\bp|.$
\section{Massive vector field}
For $m>0,$ we denote by $V^+_m$ the hyperboloid $p^2+m^2=0,\,p^0>0,$ i.e.
\be p^0=+\sqrt{\bp^2+m^2}.\label{eq:p4}\ee

Later on we will be interested in the limit $m\rightarrow 0.$ The mass hyperboloid $V^+_m$ is globally parametrized by $\bp\in \BR^3.$ In the limit $m\rightarrow 0,$ $V^+_m$ becomes the positive cone $V^+_0$, and we will remove the origin $\bp=0.$
\subsection{The tangent bundle}
Let $TV^+_m$ be the tangent bundle of $V^+_m.$
Let $p^\mu(t)$ be a (differentiable) path in $V_m^+,$ with $(p^\mu(0))=(p^0=\sqrt{\bp^2+m^2},\bp).$ For each $t$ we have
\be p^0(t)^2=\bp(t)^2+m^2.\ee
Differentiating at $t=0$ and setting $v^\mu=(dp^\mu(t)/dt)|_{t=0},$ we get
\be v^0 p^0=\bv\cdot\bp,\ee
therefore, in $TV^+_m,$ at $\bp,$  we can use only the coordinates $v^i$, the coordinate $v^0$ being given by
\be v^0(\bp)=\frac{\bv\cdot\bp}{p^0}.\label{eq:v4}\ee
\subsection{Action of the Lorentz group}
Let $\eta=(\eta_{\mu\nu})$ be the matrix $\eta=\mbox{diag}(-1,+1,+1,+1).$
The inverse matrix $\eta^{-1}=(\eta^{\mu\nu})$ has the same matrix elements. Indices are being raised and lowered with the matrices $\eta^{-1}$ and $\eta$ respectively. In particular $p_i=p^i,$ and $p_0=-p^0.$\footnote{We will never use $p_0$.}

Let $\cL$ be the orthochronous Lorentz group, that is the group of those $4\times 4$ real matrices $L=({L^\mu}_\nu) $ satisfying
$ \eta^{-1}L^T\eta=L^{-1},\, {L^0}_0>0.$

The group $\cL$ acts on $V^+_m$ via $p\mapsto Lp,$ $(Lp)^\mu={L^\mu}_\nu p^\nu.$ On the mass hyperboloid $V^+_m,$ $p^0$ is determined by $\bp.$ Therefore, on $V^+_m,$  we can write
$ (Lp)^i={L^i}_j p^j+{L^i}_0\,p^0,$
where $p^0=\sqrt{\bp^2+m^2}.$ Let $\bp\mapsto L\bp$ denote this action:
\be (L\bp)^i={L^i}_j\,p^j+(L^i)_0\,p^0.\ee
 It induces the action on the tangent bundle $TV_m^+$ as follows.

If $v^\mu$ is a vector tangent to $V^+_m$ at $\bp, $ then $(Lv)^\mu={L^\mu}_\nu v^\nu$ is tangent to $V^+_m$ at $Lp.$ Using now Eq. (\ref{eq:v4}) we obtain
\be (Lv)^i=({L^i}_j+\frac{{L^i}_0p_j}{p^0})v^j.\label{eq:Lv}\ee
We set
\be {\tD^i}_j(L,\bp)={L^i}_j+\frac{{L^i}_0p_j}{p^0}.\label{eq:Lf2}\ee
The wave functions of the massive spin $1$ particle are sections of the tangent bundle $TV^+_m.$ We denote by $\Sec(TV^+_m)$ the space of these sections. If $\bp\mapsto \f(\bp)$ is in $\Sec(TV^+_m)$, then $L\f$ is defined through the formula
\be (L\f)(L\bp)=L(\f(\bp)),\ee or, in coordinates:
\be (L\f)^i(L\bp)={\tD^i}_j(L,\bp)f^j(\bp).\ee
Replacing $L\bp$ with $\bp$ and $\bp$ with $L^{-1}\bp$, we get
\be (L\f)^i(\bp)={\tD^i}_j(L,L^{-1}\bp)f^j(L^{-1}\bp)\label{eq:Lf1}.\ee
The formula (\ref{eq:Lf1}) defines a natural linear action of $\cL$ on $\Sec(TV^+_m).$

It can be verified by a direct calculation that the matrices $\tD(L,\bp)$ satisfy the following ``cocycle relations'':
\be \tD(L_1L_2,\bp)=\tD(L_1,L_2\bp)\tD(L_2,\bp),\ee
which is equivalent to
\be (L_1L_2)\f=L_1(L_2\f),\ee
which means, that we have a linear representation of the group  $\cL$ on  $\Sec(TV^+_m).$ So far no restriction on sections $\bp\mapsto f(\bp)$ are needed. The minimal assumption is that they are (Borel) measurable. Notice also that $f(\bp)$ can be assumed to be real. As long as we are interested only in pure Lorentz transformations, and not in spacetime translations, there is no need to complexify $TV^+_m$.
\subsection{Riemannian metric on $V^+_m.$}
For $m>0$ the fibers of $TV^+_m$ carry a natural Lorentz invariant  Riemannian structure: the scalar product $ v\cdot v'=\bv\cdot\bv'-v^0{v'}^0$ is evidently Lorentz-invariant. Substituting the expression (\ref{eq:v4}) we obtain the coordinate expression for the Riemannian metric $g_{(m)ij}(\bp)$
\be g_{(m)ij}(\bp)=\delta_{ij}-\frac{p_ip_j}{\om^2+m^2}.\label{eq:g}\ee
Its inverse is given by
\be g_{(m)}^{ij}=\delta^{ij}+\frac{p^ip^j}{m^2}.\ee
Since the flat metric $\eta$ is invariant under the linear action of the Lorentz group, the induced metric $g_{m}$ is invariant under the induced action. We have
\be \tD(L,\bp)^Tg_{(m)}(L\bp)\tD(L,\bp)=g_{(m)}(\bp).\label{eq:dgd}\ee
In what follows we will need the positive square root of $g_{(m)}.$ Since, as long as $m>0,$
 $g_{(m)}$ is positive definite, there exists a unique {\em positive-definite square root\,} $h_{(m)}=({h^{(m)i}}_j)$ of $g_{(m)}.$ It can be verified that $h_{(m)}$ is given by the following explicit expression:
\be {{h_{(m)}}^i}_j(\bp)=\delta^i_j+\lambda(\bp)\pi^i\pi_j,\label{eq:hm}\ee
\be \lambda(\bp)=\frac{m}{p^0}-1.\label{eq:l1}\ee
We will also need its inverse $(h_{(m)})^{-1},$ which is given by
\be {{{h_{(m)}}^{-1}}^i}_j(\bp)=\delta^i_j+\mu(\bp)\pi^i\pi_j,\label{eq:h1}\ee
\be \mu(\bp)=\frac{p^0}{m}-1.\label{eq:h2}\ee
\begin{Remark}
We have two scalar products in $\BR^3_\bp$, the standard, Euclidean one, $\delta_{ij}$, and the one determined by the metric $g_{(m)ij}.$ In the following whenever we write the dot product or raise or lower the space index $(i,j,\ldots),$  we will always use the standard Euclidean metric.
\end{Remark}
\subsection{The Hilbert space and unitary representation}
We have
\be \sqrt{\det g_{(m)}}=\frac{m}{p^0},\ee
therefore, since $g_{(m)}$ is Lorentz invariant, $d^3p/p^0$ is a Lorentz invariant measure on $V^+_m.$ We define the Hilbert space $\tH_m$ as the Hilbert space of sections $\f(\bp)$ of $TV^+_m$ square integrable with respect to the scalar product
\be (\f,\f')_m=\int g_{(m)ij}(\bp)\bar{f}^i(\bp)f'^j(\bp)\frac{d^3p}{p^0}.\ee
By construction the formula
\be (\tU(L)\f)(\bp)\doteq (L\f)(\bp)=\tD(L,L^{-1}\bp)\f(L^{-1}\bp)\ee
defines a unitary representation of $\cL$ on $\tH_m.$

The unitary representations of $\cL$ for different values of $m$ have the same form but act in Hilbert spaces with different scalar products in the fibers. In order to be able to take the limit $m=0$ it is convenient to use just one standard fiber scalar product, independent of the value of $m$, but make the form of the representation $m$-dependent. To this end  let $\H_m$ be the Hilbert space of sections of $TV^+_m$ square integrable with respect to the standard Hermitian scalar product
\be <\f,\f'>_m=\int \bar{\f}(\bp)\cdot\f'(\bp)\frac{d^3p}{p^0},\ee
where $\bar{\f}$ denotes the complex conjugate (not needed if $\f$ is real). Then the map \be h_{(m)}:\f(\bp)\mapsto {h_{(m)}}(\bp)\f(\bp),\ee where $h_{(m)}$ is given by Eq. (\ref{eq:hm}), is an isometry from $\tH_m$ to $\H_m.$ Correspondingly we have a unitary representation $L\mapsto U_m(L)=h_{(m)}\circ \tU(L)\circ {h_{(m)}}^{-1}$ of $\cL$ on $\H_m.$ It follows then from the definition that $U(L)$ can be written as
\be (U_m(L)f)(\bp)=D_m(L, L^{-1}\bp)f(L^{-1}\bp),\ee
where
\be D_m(L,\bp)=h_{(m)}(L\bp)\tD(L,\bp)h_{(m)}^{-1}(\bp).\label{eq:dml}\ee
By using Eq. (\ref{eq:dgd}) and the definition of $h_{(m)}$ we find that the matrices $D_m(L,\bp)$ are orthogonal, which makes the property of unitarity of $U_m(L)$ evident.
\footnote{Using the terminology of \cite[p. 175]{var} one says that the cocycles $\tD$ and $D_m$ are {\it strictly cohomologous}.}

\subsubsection{The longitudinal-transversal split}
Assuming $\bp\neq 0,$ the eigenvalue equation $g_{(m)}\f=\lambda \f$ for the real symmetric matrix $g_{(m)}$ reads
\[ f^i(\bp)-\frac{\bp\cdot\f(\bp)}{\om^2+m^2}p^i=\lambda f^i(\bp),\]
i.e.
\be (1-\lambda)f^i(\bp)=\frac{\bp\cdot\f(\bp)}{\om^2+m^2}p^i.\ee
Thus either $\lambda=1,$ and then $\bp\cdot\f(\bp)=0,$ or $\lambda\neq 1,$ and then $\f(\bp)$ is proportional to $\bp.$ For each $\bp\neq 0$ let $\mP_0(\bp)$ be the orthogonal projection (in $\BR^3$ endowed with the Euclidean metric $\delta_{ij}$) on the one-dimensional subspace consisting of vectors $\bf$ proportional to $\bp$ :
\be \mP_0(\bp)\,\f(\bp)=(\bpi\cdot\f(\bp))\bpi.\ee
Then $\mP_0(\bp)$ projects onto the eigenspace of $g_{(m)}(\bp)$ belonging to the eigenvalue  $m^2/(\om^2+m^2).$
Let $\mP_1(\bp)=I-\mP_0(\bp)$ be the orthogonal projection on the complementary subspace of vectors corresponding to the eigenvalue $1$:
\be \mP_1(\bp)\f(\bp)=\f(\bp) - (\bpi\cdot\f(\bp))\bpi,\ee
so that
\be g_{(m)}(\bp)=\frac{m^2}{\om^2+m^2}\mP_0(\bp)+\mP_1(\bp).\ee
Then we immediately get
\be h_{(m)}(\bp)=\frac{m}{\sqrt{\om^2+m^2}}\mP_0(\bp)+\mP_1(\bp),\label{eq:hmp}\ee
and
\be {h_{(m)}}^{-1}(\bp)=\frac{\sqrt{\om^2+m^2}}{m}\mP_0(\bp)+\mP_1(\bp).\label{eq:hmp1}\ee
It should be noticed, however, that, as long as $m>0,$ the split above is not invariant under the action of boosts of the Lorentz group on $V_m^+.$
\subsection{The limit $m=0$}
For $m=0$ the matrices $\tD(\bp)$ are still given by Eq. (\ref{eq:Lf2}), but now $p^0=\om.$ It is less obvious that the matrices $D_m$ given by Eq. (\ref{eq:dml}) remain finite in the zero mass limit. In Proposition \ref{prop1} below we show that, for all $L$ in $\cL$ and $\bp$ in $V_m^+,$ the limit \be D_0(L,\bp)=\lim\limits_{m=0} D_m(L,\bp)\ee exists, and we provide its explicit form.
\begin{Proposition}\label{prop1}
$D_0(L,\bp)$ is finite and it is given by the following explicit formula:
\be {{D_0}^i}_j(L,\bp)={L^i}_j+
{L^i}_0\pi_j-
{\pi'}^i{L^0}_j+
{\pi'}^i\pi_j(1-{L^0}_0),\label{eq:d0}\ee
where $\bpi'=\frac{L\bp}{|L\bp|}.$
For each $L\in\cL$ and $\bp\neq 0$ the matrix $D_0(L,\bp)$ is orthogonal
\be D_0(L,\bp)^T=D_0(L,\bp)^{-1}.\label{eq:dort}\ee
The representation
\be (U_0(L)\f)(\bp)\doteq D_0(L,L^{-1}\bp)\f(L^{-1}\bp)\label{eq:u0}\ee
of $\cL$ on the Hilbert space $\H_0=L^2(\BR^3_\bp,d^3p/\om)\otimes\BC^3$
is unitary.
\end{Proposition}
\begin{proof}
Substituting Eqs. (\ref{eq:hmp}) and (\ref{eq:hmp1}) into Eq. (\ref{eq:dml}) we obtain for $D_m(L,\bp)$ the sum of four terms $t_1,t_2,t_3,t_4,$ where
\begin{eqnarray} t_1(m)&=&\frac{\sqrt{|\bp|^2+m^2}}{\sqrt{|L\bp|^2+m^2}}\mP_0(L\bp)\tD(L,\bp)\mP_0(\bp),\\
t_2(m)&=&\mP_1(L\bp)\tD(L,\bp)\mP_1(\bp),\\
t_3(m)&=&\frac{m}{\sqrt{|L\bp|^2+m^2}}\mP_0(L\bp)\tD(L,\bp)\mP_1(\bp),\\
t_4(m)&=&\frac{\sqrt{|\bp|^2+m^2}}{m}\mP_1(L\bp)\tD(L,\bp)\mP_0(\bp).
\end{eqnarray}
Noticing that the orthogonal projections $\mP_0$ and $\mP_1$ do not depend of $m$ we see that $t_1(0)$ and $t_2(0)$ are finite, and are given by
\be t_1(0)=\frac{|\bp|}{|L\bp|}\mP_0(L\bp)\tilde{D_{(0)}}(L,\bp)\mP_0(\bp),\ee
\be t_2(0)=\mP_1(L\bp)\tilde{D_{(0)}}(L,\bp)\mP_1(\bp),\ee
while the third term vanishes
\be t_3(0)=0.\ee
We will now show that, surprisingly,  the fourth term $t_4(m)$ also vanishes in the limit $m=0.$ To this end we consider first
the product
\be D\mP_0(m)\doteq\tD(L,\bp)\mP_0(\bp)\ee of its last two factors. Using the definitions we have
\be
{{D\mP_0(m)}^i}_k=\left({L^i}_j+\frac{{L^i}_0p_j}{p^0}\right)\frac{p^jp_k}{|\bp|^2}
=\left(\frac{{L^i}_jp^j}{|\bp|^2}+\frac{{L^i}_0}{p^0}\right)p_k.
\ee
Setting $p'=Lp,$ we have ${L^i}_jp^j=p'^i-{L^i}_0p^0,$ therefore
\be {{D\mP_0(m)}^i}_k=\frac{p'^ip_k}{|\bp|^2}+{L^i}_0\left(\frac{1}{p^0}-\frac{p^0}{|\bp|^2}\right)p_k=
\frac{p'^ip_k}{|\bp|^2}-m^2\frac{{L^i}_0p_k}{p^0|\bp|^2}.\label{eq:dp0m}\ee
Now, the projection $\mP_1(\bp')$ vanishes on the first term $\frac{p'^ip_k}{|\bp|^2},$ as the range of this matrix is in the longitudinal subspace at $\bp'$, and $\frac{\sqrt{|\bp|^2+m^2}}{m}\mP_1(L\bp)$ acting on the second term vanishes linearly in $m.$ Therefore also $t_4(0)=0,$ and so
\be D_0(L,\bp)=\frac{|\bp|}{|\bp'|}\mP_0(\bp')\tilde{D_{(0)}}(L,\bp)\mP_0(\bp)+\mP_1(\bp')\tilde{D_{(0)}}(L,\bp)\mP_1(\bp).\label{eq:d0l1}\ee
Since the subbundle of longitudinal vectors is invariant under the action of $\tilde{D_{(0)}},$ we have
$\mP_1(\bp')\tilde{D_{(0)}}(L,\bp)\mP_0(\bp)=0$ and $\mP_0(\bp')\tilde{D_{(0)}}(L,\bp)\mP_0(\bp)=\tilde{D_{(0)}}(L,\bp)\mP_0(\bp).$ Therefore Eq. (\ref{eq:d0l1}) simplifies to
\be D_0(L,\bp)=\tilde{D_{(0)}}(L,\bp)+\frac{|\bp|}{|\bp'|}\tilde{D_{(0)}}(L,\bp)\mP_0(\bp)-\mP_0(\bp')\tilde{D_{(0)}}(L,\bp).\ee
The matrix $\tilde{D_{(0)}}(L,\bp)\mP_0(\bp)$ we have already calculated - we just set $m=0$ in Eq. (\ref{eq:dp0m}). Using now the fact that $L^T\eta L=\eta,$ a straightforward algebra gives us \be {(\mP_0(\bp')\tilde{D_{(0)}}(L,\bp))^i}_j={\pi'}^i({L^0}_j+\pi_j{L^0}_0),\ee and leads to the formula (\ref{eq:d0}) of Proposition \ref{prop1}.\footnote{Notice that $\tilde{D_{(0)}}(L,\bp)$ gives the first two terms of the right hand side of (\ref{eq:d0}).} The matrices $D_0(L,\bp)$ are orthogonal as limits of a continuous family of orthogonal matrices, and, since the measure $d^3p/p^0$ is Lorentz invariant, the representation $U_0$ is unitary.
\end{proof}
We notice that for pure space rotations: ${L_i}^0={L_0}^j=0,{L_0}^0=1,$ we have ${{D_0}_i}^j={L_i}^j.$
\subsection{The Poincar\'e group Lie algebra}
So far we have discussed only the homogeneous transformations form the Poincar\'e group - the Lorentz transformations from $\cL.$ Now we add translations. They are implemented by complex phase rotations:
\be (U_0(a)\f)(\bp)=e^{ia\cdot p}\f(\bp)=e^{i(a^1p^1+a^2p^2+a^3p^3-a^0p^0)}\f(\bp).\ee
These transformations are evidently unitary on $\H_0.$
Having the unitary representation of the whole group we will calculate now the self-adjoint infinitesimal generators. For translations we define
$P_\mu=-i(dU(a)/da^\mu)|_{a=0},$ obtaining
\be P^i\f(\bp)=p^i\f(\bp),\quad P^0\f(\bp)=|\bp|\f(\bp).\ee

The Lie algebra of the orthogonal group $SO(4)$ consists of antisymmetric matrices $\tilde{m}_{\mu\nu}$ given by
\be (\tilde{m}_{\mu\nu})^{\sigma\rho}=\delta_\mu^\sigma\delta_\nu^\rho-\delta_\nu^\sigma\delta_\mu^\rho.\ee
The matrices $m^{\mu\nu}=\tilde{m}^{\mu\nu}\eta$ form then a basis in the Lie algebra of the Lorentz group
\be {(m_{\mu\nu})^\sigma}_\rho= \delta_\mu^\sigma\eta_{\nu\rho}-\delta_\nu^\sigma\eta_{\mu\rho}.\label{eq:munu}\ee
Defining ${\bf m}=(m^i)$ and ${\bf n}=(n^i)$ through $m^i=\frac12 \epsilon^{ijk}m_{jk},\, n^i=m_{0i},$ we obtain the commutation relations
\be [m^i,m^j]=-\epsilon_{ijk}m^k,\,[m^i,n^j]=-\epsilon_{ijk}n^k,\, [n^i,n^j]=\epsilon_{ijk}m^k.\ee
Then the densely defined self-adjoint operators $M^i,N^i$ are given by
\be M^i=-idU(\exp(t\, m^i))/dt|_{t=0},\, N^i=-idU(\exp(t\, n^i))/dt|_{t=0}.\ee
Using our explicit formulas,  we obtain:
\begin{eqnarray}
{\bf M}&=&{\bf L}+{\bf s},\\
{\bf N}&=&{\bf K}+{\bf k},\label{eq:N}
\end{eqnarray}
where
\be {\bf L}=-i\bp\times \partial/\partial\bp,\quad
{\bf K}=i\om\partial/\partial\bp,\quad  {\bf k}=\bpi\times {\bf s},\ee
\be {\bf s}=-i{\bf m}.\label{eq:s}\ee
The commutation relations obtained from the definitions are the standard ones:
\be [N^i,P^j]=i\delta^{ij}P^0,\,
[N^i,P^0]=iP^i,\,
[M^i,P^j]=i\epsilon_{ijk}P^k,\,
 [M^i,P^0]=0,\ee
\be [N^i,N^j]=-i\epsilon_{ijk}M^k,\,
[M^i,N^j]=i\epsilon_{ijk}N^k,\,
[M^i,M^j]=i\epsilon_{ijk}M^k.\ee
Replacing $\bf{M}$ by $\bf{L}$ and $\bf{N}$ by $\bf{K}$ we obtain the same commutation relations.
\begin{Remark}\label{rem:2}

The representation $U_0$ defined by (\ref{eq:u0})
is unitary on the Hilbert space $\H_0$ of vector-valued functions square integrable with respect to the scalar product
\be <\f,\f'>_0=\int \f^\dagger(\bp)\f'(\bp)\frac{d^3p}{\om}.\label{eq:ssp}\ee
In Sec. \ref{sec:pp} we will use a different scalar product using a non-Lorentz-invariant measure $d^3p.$
To make the representation $U_0$ unitary with respect to this scalar product, we need to adjust $D_0$
by introducing an extra scaling factor and defining
\be D(L,\bp)=\sqrt{\frac{\om}{|L\bp|}}D_0(L,\bp).\ee
The expression for infinitesimal generators for this unitary representation remains the same, except that in the definition of ${\bf K}$ we must replace $i\om\partial/\partial\bp$ by $i\om\partial/\partial\bp+\frac12 \bpi.$

Notice however, that with the scalar product (\ref{eq:ssp}), and with $D_0$ replaced by $D$, the transversal components of $\f(\bp)$ cannot be {\it directly\,} interpreted as tangent vectors, since they do not have the correct transformation properties under boosts.
\end{Remark}
The unitary space inversion operator $\Pi$ and antiunitary time inversion operator $\Theta$ for this representation  are given by
\be (\Pi{\bf f})(\bp)={\bf f}(-\bp),\quad (\Theta{\bf f})(\bp)={\bf f}^*(-\bp).\ee
For a general representation of the Poincar\'{e} group one defines the four-dimensional Pauli-Lubanski pseudovector $W^\mu$ as
\be W_\mu=\frac12 \epsilon_{\nu\rho\sigma\mu}\,P^\nu M^{\rho\sigma}.\label{eq:pl}\ee
We have
\be W^0={\bf P}\cdot{\bf M},\, {\bf W}=P^0\,{\bf M}-{\bf P}\times {\bf N}.\label{eq:w0w}\ee

It follows from the very definition that $\eta_{\mu\nu}P^\mu W^\nu=0.$  For mass zero representations $P^\mu$ is lightlike, and therefore $W^\mu$ is proportional to $P^\mu$:
\be W^\mu=\Lambda P^\mu.\label{eq:hel}\ee
The proportionality operator $\Lambda$ commutes with all the generators and is called the helicity operator (see e.g. \cite[p.64]{ryder}).
In our case one finds that for generators $P^\mu,\bf{L},\bf{K}$ we have $W^\mu=0$, therefore $\Lambda=0,$ while for the generators $P^\mu,\bf{M},\bf{N}$ we have
\be \Lambda=\bpi \cdot\bM=\bpi \cdot\bf{s}, \label{eq:lambda}\ee
or explicitly, using Eq. (\ref{eq:s}),
\be(\Lambda \f)(\bp)
=i\bpi\times \f(\bp) .\ee The spectrum of
$\Lambda$ is discrete and consists of three points $\lambda=\pm 1$
and $0.$ Therefore $\Lambda^2$ is a projection onto the subspace $\H_{ph}$ of $\H_0.$ $\H_{ph}$ is a direct sum of eigenspaces
$\H_{\pm }$ of $\Lambda$ corresponding to eigenvalues
$\lambda=\pm
1.$ The photon states are represented by vectors in $\H_{ph}.$ The
orthogonal complement $\H_l$ of $\H_{ph}$ in $\H_0$ describes a
spinless particle. It follows from these definitions that

\be
\H_{ph}=\{\f\in\H_0:\,\bp\cdot\f (\bp )=0\}\ee
and that
\be \H_l=\{\f\in\H_0:\,\f (\bp )=c(\bp)\bp\} \ee
for some scalar function $c(\bp).$ Since $\Lambda$ commutes with all the generators of the Poincar\'e group, they leave the subspaces $\H_{ph}$ and
$\H_l$ invariant.
\section{Application of the explicit form: photon polarization vectors boosted to the speed of light}
\subsection{Action of pure boosts}\label{boosts}
We will be interested in one-parameter subgroups of Lorentz transformations $L(\n,s),\\ \n\in\BR^3,\n^2=1,\,s\in\BR,$ defined
by
\be L(\n,s)=\exp(-s\, \n\cdot {\bf n})=\exp\left(s\left(\begin{smallmatrix}-w_1&0&0&0\\-w_2&0&0&0\\-w_3&0&0&0\\0&-w_1&-w_2&-w_3\end{smallmatrix}\right)\right).\ee
For the matrix ${L(\n,s)^\mu}_\nu$ we then obtain\footnote{The parameter $s$ is known as {\it rapidity\,}, $s=\mbox{atanh}(\beta),$ $\beta=v/c.$}
\begin{eqnarray}
{L^0}_i &=&{L^i}_0=-w_i\sinh(s),\notag\\
{L^i}_j&=&\delta^i_j+w^iw_j(\cosh(s)-1),\label{eq:lij}\\
{L^0}_0&=&\cosh(s).\notag\end{eqnarray}
For a fixed $\n$ we have the group property:
\be L(\n,s_1)L(\n,s_2)=L(\n,s_1+s_2).\ee
In particular we have $L(\n,s)^{-1}=L(\n,-s).$\\
\noindent For some special sections $\f(\bp)$ of $TV_0^+$ we will be interested in calculating the limits
\be \lim\limits_{s\rightarrow\pm\infty}U_0(L(\n,s))\f.\ee
Having in mind Eq. (\ref{eq:u0}) let \be D_0'(\n,s,\bp)\doteq D_0(L(\n,s),L(\n,s)^{-1}\bp).\ee
The following Proposition shows that the limits  \be
D_0'(\n,\bp)_\pm=\lim\limits_{s\rightarrow \pm\infty}D_0'(\n,s,\bp)\ee
exist, and provides their explicit form.
\begin{Proposition}
For all $\bp$ not parallel to $\n,$ we have
\be D_0'(\n,\bp)_\pm=
\mathbbm{1}\pm
\bpi\n^T-
\frac{\n\n^T
-(\bpi\cdot\n)\bpi\n^T
+\bpi\bpi^T\pm\n\bpi^T}{1\pm \bpi\cdot\n}.\label{eq:dp}\ee
where for any two vectors ${\bf a},{\bf b}$ we denote by ${\bf a}{\bf b}^T$ the matrix with components $a_ib^j.$
\label{prop:d}\end{Proposition}
\begin{proof}
The proof is by a straightforward (though tedious) algebra. From the definition of $D'_0$ and Eq. (\ref{eq:d0}) we have the following explicit expression for $D_0(L,L^{-1}\bp)$:
\be {{D_0(L,L^{-1}\bp)}^i}_j(L,\bp)={L^i}_j+
{L^i}_0\pi'_j-
{\pi}^i{L^0}_j+
{\pi}^i\pi'_j(1-{L^0}_0),\label{eq:d0p}\ee
where now $\bpi'=\frac{L^{-1}\bp}{|L^{-1}\bp|}.$
Substituting there $L(\n,s)$ for $L,$ we obtain fractions containing  $\cosh(s), \sinh(s),$ their squares and products. The terms with  $\sinh(s)\cosh(s)$ cancel out, while the terms with $\cosh^2(s), \sinh^2(s)$ collect so that we can apply the identity $\cosh^2(s)-\sinh^2(s)=1$. Dividing the numerators and denominators (which are now linear in $\cosh(s)$ and $\sinh(s)$)  by $\cosh(s)$ we use the fact that
$\lim\limits_{s\rightarrow \pm \infty} \tanh(s)=\pm 1$ to obtain the result.
\end{proof}
\subsection{The polarization basis}
Using the Euclidean metric we will silently rise and lower the space indices with the Kronecker deltas $\delta^{ij}$ and $\delta_{ij}.$

Let $\m$ be a unit ($\m\cdot\m=1$) vector in $\BR^3$, and let $\bp$ be a non-zero vector in $\BR^3$, $\bp$ not parallel to $\m.$ Define the following two vectors $e^\m_1(\bp),e^\m_2(\bp)$ in $\BR^3$  as\footnote{Cf. Ref. \cite[Eq. (3.3)]{bb1}}
\begin{eqnarray}
e^\m_1(\bp)&=&\frac{\bp\times\m}{\vert \bp\times\m\vert},\\
e^\m_2(\bp)&=&\frac{\bp\times(\bp\times \m)}{\vert\bp\times(\bp\times \m)\vert}\label{eq:em2}.
\end{eqnarray}
Adding the third vector field $e_\parallel$ defined as
\be e_\parallel(\bp)=\bpi,\ee
we obtain an orthonormal basis, which can be considered as an orthonormal moving frame in the (real) tangent bundle of the positive light cone in momentum space. The sections $e^\m_1$ and $e^\m_2$ of the tangent bundle $TV^+_0$ become singular on the straight line determined by $\m.$
\begin{Remark}
One can check that $\Lambda e_\parallel(\bp)=0,$ and that the states $e^\m_1(\bp)\pm i e^\m_2(\bp)$ are eigenstates of the helicity operator to the eigenvalue $\pm 1:$
\be \Lambda  (e^\m_1\pm i e^\m_2)=\pm (e^\m_1\pm i e^\m_2).\ee
\end{Remark}
\begin{Proposition}\label{prop2}
The limits ${e^\m_\alpha}^\pm=\lim\limits_{s\rightarrow \pm\infty} U_0(L(\n,t))\,e^\m_\alpha,\,\alpha=1,2,$ exist, and are given by
\be {e^\m_1}^\pm= \frac{\pm 1}{\sqrt{1-(\n\cdot\m)^2}}\left(\n\times\m-\frac{\bpi\cdot(\n\times\m)}{1\pm\bpi\cdot\n}(\bpi\pm\n)\right),\ee
\be {e^\m_2}^\pm=  \frac{1}{\sqrt{1-(\n\cdot\m)^2}}\left(\n\times(\n\times\m)-\frac{\bpi\cdot(\n\times(\n\times\m))}{1\pm\bpi\cdot\n}(\bpi\pm\n)\right).\ee
Moreover, the vector fields $e_\m^\alpha$ are invariant under the boosts in $\n$-direction. We have
\be (\n\cdot{\bf N})\, e^\m_\alpha=0,\ee
where ${\bf N}$ is the boost generator (\ref{eq:N}).
\end{Proposition}
\begin{proof}
We know from Proposition \ref{prop:d} that the limits $D_0'(\n,\bp)_\pm$ exist. On the other hand, using the definitions of $e^\m_\alpha$ and $L(\n,s)^{-1}=L(\n,-s),$ as well as the fact that $\lim\limits_{s\rightarrow \pm\infty} \tanh(s)=\pm 1,$ we obtain
\be \lim\limits_{s\rightarrow\pm\infty} e^\m_1(L(\n,s)^{-1}\bp)=\frac{\pm\n\times\m}{\sqrt{1-(\n\cdot\m)^2}},\label{eq:nm}\ee
\be \lim\limits_{w\rightarrow\pm\infty} e^\m_2(L(\n,w)^{-1}\bp)=\frac{\n\times(\n\times\m)}{\sqrt{1-(\n\cdot\m)^2}}.\label{eq:nnm}\ee
The result follows then by an application of the matrix of $D_0'(\n,\bp)_\pm$ given by formula (\ref{eq:dp}) to the vectors (\ref{eq:nm}) and (\ref{eq:nnm}). The last statement of the proposition is the immediate consequence of the definitions.
\end{proof}
\subsubsection{Lorentz-boost eigenmodes}\label{sec:ex}
Let us take for $\n$ the vector
$$ \n=(0,0,1),$$ (boost in the direction of the third axis), and for $\m$ the vector
$${\bf a}=(0,1,0).$$ We consider the case of $\lim\limits_{w\rightarrow -\infty}.$ Using the formulas (\ref{eq:nm}) and (\ref{eq:nnm}) we obtain the following two mutually orthogonal unit vector fields:
\be
e_1(\bp)=\frac{1}{1-\pi^3}\begin{pmatrix}1-\pi^3-(\pi^1)^2\\-\pi^1\pi^2\\ (1-\pi^3)\pi^1\end{pmatrix},\quad
e_2(\bp)=\frac{1}{1-\pi^3}\begin{pmatrix}\pi^1\pi^2\\-(1-\pi^3)+(\pi^2)^2\\- (1-\pi^3)\pi_2\end{pmatrix}.\label{eq:e12}\ee
These two vector fields are transversal; they define a photon polarization basis (notice that while each of the original vector fields $e^\m_1,e^\m_2$ has two singular points on the unit sphere, in their speed-of-light limits these two singularities merge into just one---the minimum required by the ``hairy ball theorem'' of algebraic topology).  Together with the third vector~field

\vspace{-3pt}
\be e_3(\bp)=\bpi,\label{eq:e3}\ee
which is longitudinal and spans the helicity-zero subspace, they form an orthonormal basis in $T_\bp V_0^+$---cf. Fig. \ref{fig0}.
 \begin{figure}[!ht]
 \includegraphics[width=12.5cm, keepaspectratio=true]{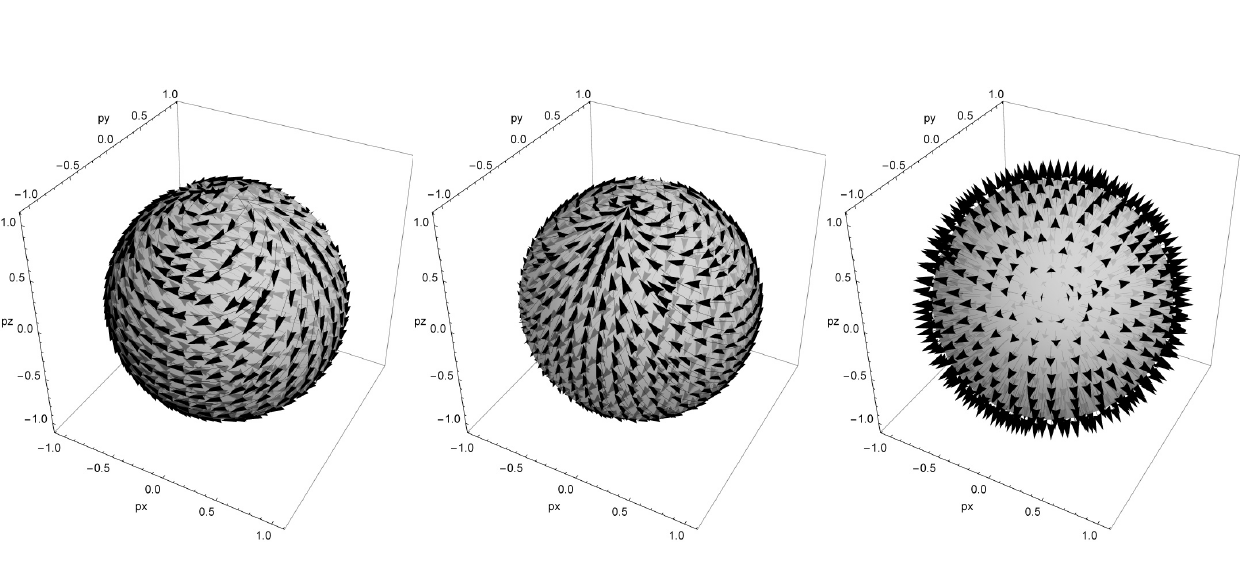}
 \caption{The vector fields $e_1(\bp),e_2(\bp),e_3(\bp)$ plotted at the unit sphere in momentum space.}\label{fig0}
\end{figure}

It is easily verified that
\be M^3 e_1=-i e_2,\, M^3 e_2=i e_1,\, M^3 e_3=0,\label{eq:m3}\ee
and
\be N^3 e_1=N^3 e_2=N^3 e_3=0.\label{eq:n30}\ee
\begin{Remark}The vector fields $e_1,e_2$ are real. Introducing their complex combinations
\be e_\pm=\frac{1}{\sqrt{2}}(e_2\pm i e_1),\ee
we obtain two unit complex vector fields of definite helicity and the third component of angular momentum:
\be M^3 e_\pm=\mp e_\pm,\, \Lambda e_\pm=\pm e_\pm.\ee
Evidently we also have
\be N^3 e_\pm=0.\ee
\end{Remark}
The basis $e_\alpha,$ is naturally embedded into a one-parameter family of Lorentz-boost eigenmodes\footnote{For the scalar wave equation some of their properties have been discussed in Ref. \cite{bliokh}.} $e^{(\lambda)}_\alpha$ defined by the formula below:
\be e^{(\lambda)}_\alpha(\bp)=\exp\left(i \lambda \log\sqrt{ \frac{|\bp|-p_3}{|\bp|+p_3}}\right)\,e_\alpha(\bp),\, (\alpha=1,2,3),\,\lambda\in\BR.\ee
A straightforward verifications shows that $e^{(\lambda)}_\alpha$ still satisfies Eq. (\ref{eq:m3}), but Eq. (\ref{eq:n30}) is replaced by
\be N^3 e^{(\lambda)}_\alpha=\lambda\,e^{(\lambda)}_\alpha.\label{eq:n3l}\ee
\begin{Remark}
In fact one can show that every solution of the eigenvalue problem
\be N^3 f=\lambda f\ee
is of the form
\be f(\bp)=\sum_{\alpha=1}^3 c_\alpha(p_1,p_2)e^{(\lambda)}_\alpha(\bp).\ee
\end{Remark}
\subsection{The teleparallel connection\label{sec:tpc}}
In the following we will use the notation, assumptions and results of Sec. \ref{sec:ex}.
\subsubsection{Stereographic coordinates on the lightcone\label{sec:sc}}
 We will use the moving frame $e_\alpha(\bp),\,(\alpha=1,2,3)$ on $V_0^+$ to define teleparallel affine connection on $V_0^+.$ The frame has a singularity at the points $\bp=(0,0,t),\, t>0,$  and the natural rotational covariance with respect to rotations around the third axis.
It takes especially simple form in a coordinate system $X=(x,y,\omega)$ using the stereographic projection from the unit sphere in the momentum space. \footnote{Stereographic projection coordinates are also being used, in a similar context, in Refs. \cite{staru1},\cite{staru2}, but the coordinates $x,y$ in these papers are twice ours $x,y$, as the projection plane in these papers is positioned at the bottom of the unit sphere, and not through the center.}

We set
\begin{eqnarray}
p^1&=& \frac{2x\omega}{x^2+y^2+1},\\
p^2&=& \frac{2y\omega}{x^2+y^2+1},\\
p^3&=& \frac{x^2+y^2-1}{x^2+y^2+1}\,\omega,
\end{eqnarray}
and the inverse transform
\begin{eqnarray}
x&=&\frac{\pi^1}{1-\pi^3},\\
y&=&\frac{\pi^2}{1-\pi^3},\\
\omega&=&|\bp|.
\end{eqnarray}
Given a coordinate system $X^i$ and a vector $\xi$ with coordinates $\xi^i,$ we have the standard transformation law to another coordinate system $X^{i'}$:
\be \xi^{i'}=\frac{\partial X^{i'}}{\partial X^i}\xi^i.\ee
Applying this law to the vectors $e_1,e_2, e_3$ we obtain their components in stereographic coordinates $\omega,x,y$;\footnote{Surprisingly they happen to essentially coincide with the basis $\vec{E}_1,\vec{E}_2,\vec{E}_3$
considered in Ref. \cite[Eq. (2.42)]{dob1}: $e_1=-\vec{E}_1,\, e_2=-\vec{E}_2,\,e_3=\vec{E}_3.$ In fact we have $e_1=\widehat{\partial_x},\,e_2=-\widehat{\partial_y},\,e_3=\widehat{\partial_\omega}=\partial_\omega,$ where the hat over a vector denotes the unit vector in its direction.}
\be e_1=\begin{pmatrix}\frac{1+x^2+y^2}{2\omega}\\0\\0\end{pmatrix},\,
e_2=\begin{pmatrix}0\\-\frac{1+x^2+y^2}{2\omega}\\0\end{pmatrix},\,
e_3=\begin{pmatrix}0\\0\\1\end{pmatrix}.\label{eq:ea}\ee
The Euclidean metric in stereographic coordinates has the form
\be g=(g_{ij})=
\begin{pmatrix}
\frac{4\omega^2}{(1+x^2+y^2)^2}&0&0\\
0&\frac{4\omega^2}{(1+x^2+y^2)^2}&0\\
0&0&1
\end{pmatrix},\ee
\be
g^{-1}=(g^{ij})=
\begin{pmatrix}
\frac{(1+x^2+y^2)^2}{4\omega^2}&0&0\\
0&\frac{(1+x^2+y^2)^2}{4\omega^2}&0\\
0&0&1\\
\end{pmatrix},
\ee
while the Lorentz invariant degenerate metric ${g_0}_{ij}=\delta_{ij}-\pi_i\pi_j$ is obtained from $g_{ij}$ by replacing $1$ with $0$ in the right bottom corner.
We can use the metric $g_{ij}$ to obtain the corresponding dual basis  $e^1,e^2,e^3$ in the cotangent bundle:
\be e^1=\begin{pmatrix}\frac{2\omega}{1+x^2+y^2}\\0\\0\end{pmatrix},\,
e^2=\begin{pmatrix}0\\ -\frac{2\omega}{1+x^2+y^2}\\0\end{pmatrix},\,
e^3=\begin{pmatrix}0\\0\\1\end{pmatrix}.\ee
The Lorentz invariant volume form $vol =d^3p/|\bp|$ becomes
\be vol = \frac{4\omega}{(1+x^2+y^2)^2}\,d\omega\,dx\,dy.\ee
\subsubsection{The connection coefficients}
The moving frame $e_\alpha$ defines a unique flat affine connection on $V_0$ in which the vector fields $e_\alpha$ are parallel:
\be \nabla_i e_\alpha^k\doteq\partial_i e_\alpha^k+\Gamma_{ij}^k\, e_\alpha^j=0.\ee
The coefficients $\Gamma_{ij}^k$ of this teleparallel connection,   in the coordinate system $(X^i)=(x,y,\omega)$ are then given by
\be {\Gamma_i}_j^k=e_\alpha^k\frac{\partial e_j^\alpha}{\partial X^i},\ee
where $e^\alpha$ is the dual basis, thus ${e_i}^\alpha {e^j}_\alpha=\delta^j_i.$

A straightforward calculation gives then the following expressions:
\be
\Gamma_1=\frac{-2x}{a}\begin{pmatrix}1&0&0\\0&1&0\\0&0&0\end{pmatrix},\,
\Gamma_2=\frac{-2y}{a}\begin{pmatrix}1&0&0\\0&1&0\\0&0&0\end{pmatrix},
\Gamma_3=\frac{1}{\omega}\begin{pmatrix}1&0&0\\0&1&0\\0&0&0\end{pmatrix},\,
\label{eq:gg}\ee
where $a=1+x^2+y^2.$ Our connection $\Gamma$ has the properties $\nabla_i g =\nabla_i g_0=0,$ and $\nabla_i vol=0,$ therefore it should have come under the scope of affine connections discussed by Staruszkiewicz in Refs. \cite{staru1,staru2}.
However, as we show in the paragraph below, it  is not `semi-symmetric' -- an extra condition imposed on the class of connections analyzed by Staruszkiewicz.

\subsubsection{The connection $\Gamma_{ij}^k$ given by Eq. (\ref{eq:gg})  is not semi-symmetric.}Let $M$ be an $n$-dimensional manifold with a coordinate system $x^i.$ A linear connection $\nabla$ with connection coefficients $\Gamma_{ij}^k$ is said to be semi-symmetric if its torsion tensor $T_{ij}^k=\Gamma_{ij}^k-\Gamma_{ji}^k$ is of the form \be T_{ij}^k=\delta_i^k \tau_j-\delta_j^k \tau_i,\ee  $\tau$ being a $1$-form. If that is the case, contracting the indices $i,k$ we get $\tau_j=\frac{1}{n-1}T_{ij}^i.$ In our case $n=3$, therefore  \be \tau_j=\frac12 T_{ij}^i.\ee For the connection given by Eq. (\ref{eq:gg}) we get, for instance, $\tau_2=-y/(1+x^2+y^2).$   But for a semi-symmetric connection we should have, for instance, $T_{12}^1=\delta_1^1 \tau_2-\delta_2^1 \tau_1=\tau_2,$
while from Eq. (\ref{eq:gg}) we have $T_{12}^1=\Gamma_{12}^1-\Gamma_{21}^1=0-0=0.$ Thus our connection is not semi-symmetric.
\section{Photon position operator with commuting components and axial symmetry \label{sec:pp}}
To discuss the photon localization it is more convenient to work in representation in which the scalar product in the Hilbert space of sections of the tangent bundle $TV_0^+$ is given by the formula (c.f. Remark \ref{rem:2}.)
\be (f,f')=\int_{\BR^3}f(\bp)^\dagger f'(\bp)d^3p.\label{eq:sc0}\ee
Since the measure $d^3p$ is not Lorentz-invariant the formula for the boost operator gets now an extra term (comparing to Eq. (\ref{eq:N}), and takes the form
\be {\bf N}={\bf K}+\frac{i}{2}\bpi+{\bf n}.\label{eq:N0}\ee
As a consequence the sections $\bp\mapsto e_{\alpha}(\bp)$ do not any longer satisfy Eq. (\ref{eq:n30})  - but $\bp\mapsto |\bp|^{-1/2}e_{\alpha}(\bp)$ do.

\noindent {\bf In what follows we will be using the scalar product (\ref{eq:sc0}). We will denote $\H$ the corresponding Hilbert space: $\H=L^2(\BR^3,d^3p)\otimes\BC^3.$}

\noindent The map $f\mapsto |\bp|^{1/2}f$ is an isometry between the two Hilbert spaces, the space $\H$ with the scalar product defined with the measure $d^3p,$ and $\H_0,$ with the scalar product $d^3p/p_0.$

\subsection{Position operators as covariant derivatives}
There is a straightforward relation between the position operator and a covariant derivative concept in vector bundles. Indeed, the main property required from any position operator $Q_i$ is the property of satisfying the canonical commutation relations with the momentum operators $P_i$:

\be [Q_j,P^k]=i\delta_j^k\,I.\ee

In our case, it implies that for any section $\f(\bp)$ of $T(V_0^+)$ and any scalar function $\phi(\bp)$ we have
\be -i\,Q_j(\phi\f)(\bp)=\frac{\partial \phi(\bp)}{\partial p_j}\,\f(\bp)+\phi(\bp)(Q_j\f)(\bp).\label{eq:lei}\ee
Thus the operators $-iQ_i$ have the Leibniz rule property, the main property defining a covariant derivative $\nabla_i$ in a vector bundle (cf. e.g. \cite[p. 89, Eq. (1.1)]{hermann}). For $Q_i$ to be Hermitian, $\nabla_i$ must be anti-Hermitian, and for this to be the case, the (linear) connection determined by $\nabla_i$  should preserve the fiber scalar product $(\f(\bp),\f'(\bp))_\bp=\f(\bp)^\dagger \f'(\bp).$ Then \be {\bf Q}=i\boldsymbol{\nabla},\label{eq:qid}\ee or, explicitly
\be (Q_j\f)^k(\bp)=i\partial_j f^k(\bp)+i\Gamma_{il}^k(\bp)f^l(\bp).\label{eq:Qd}\ee
\subsection{The Pryce connection and operator - geometric construction}
 There is a standard construction in the differential geometry of vector bundles that results in a canonical connection adapted to a split of a trivial vector bundle into a direct sum of its two vector sub-bundles (cf. e.g. \cite[p. 319, 4.]{greub}, \cite[Exercise 10]{crainic}). We adapt this standard construction to our purpose as follows. First, using the global coordinates $p^i$,  we realize the tangent bundle $TV_0^+$ as a trivial product bundle $V_0^+\times \BR^3.$ In the trivial bundle we have a canonical covariant derivative $\partial_i=\frac{\partial}{\partial p^i}.$
 Our trivial bundle is naturally split into a direct sum of the helicity zero sub-bundle and the helicity $\pm 1$ sub-bundle of photon states. Let $P(\bp)$ denote the orthogonal projection on the helicity zero states
 \be (P(\bp)f)^i(\bp)=\pi^i\pi_jf^j(\bp).\label{eq:P}\ee
 Then $I-P$ is the orthogonal projection on the complementary sub-bundle of photon states.
 The natural covariant derivative {\it adapted to this splitting} is then defined by the formula\footnote{Cf. \cite[Eq. 2]{jj}. We use the label ${}^{\mbox{PR}}$ to mean either ``Projection" or ``Pryce".}
 \be \nabla_i^{\mbox{PR}}=P d_i P+(I-P)\partial_i(I-P).\label{eq:cd}\ee
 Using the idempotent property $P^2=P,$ we find a simpler form:
 \be \nabla_i^{\mbox{PR}}=\partial_i+[P,\partial_i(P)],\label{eq:pdp}\ee
 while substituting the explicit form (\ref{eq:P}) of $P$ leads to
 \be (\nabla_i^{\mbox{PR}}\f)^j=\partial_if^j+\frac{1}{|\bp|}\left(\pi^j\delta_{ik}-\pi_k\delta_i^j\right)f^k,\label{eq:gpr}\ee
 thus the corresponding connection coefficients are given by
 \be {\Gamma^{\mbox{PR}}}_{ik}^j={{[P,\partial_i(P)]}^j}_k=\frac{1}{|\bp|}\left(\pi^j\delta_{ik}-\pi_k\delta_i^j\right).\label{eq:Gpr}\ee
\subsubsection{The Pryce connection is metric semi-symmetric}
A {\it metric connection\,} is semi-symmetric - cf. e.g. \cite{yano,yilmaz} and references therein - if its connection coefficients $\Gamma_{ij}^k$ are of the form
\be  \Gamma_{ik}^j=
{\overline{\Gamma}}_{ik}^j
 +\delta_i^j \tau_k-g_{ik}\tau^j\label{eq:gss},\ee
where ${\overline{\Gamma}}_{ik}^j$ are the coefficients os the Levi-Civita connection of the metric,  and where $\tau_i$ and $\tau^i=g^{ij}\tau_j$ are covariant and contravariant components of a vector field, respectively. In that case the curvatures of the two connections are related by\footnote{For the curvature tensor components of a connection $\nabla$ we use the convention
$ {{R_{ij}}^k}_l\xi^l=\left((\nabla_i\nabla_j-\nabla_j\nabla_i)\xi\right)^k.$}
\be {{R_{ij}}^k}_l={{\overline{R}_{ij}}^k}_l+\delta_j^k\tau_{il}-\delta_i^k\tau_{jl}+\delta_{il}{\tau_j}^k-\delta_{jl}{\tau_i}^k,\ee
where
\be \tau_{ij}=\nabla_i\tau_j-\tau_i\tau_j+\frac12 g_{ij}\tau^2,\ee
and $\tau^2=g^{ij}\tau_i\tau_j.$

In our case $g_{ij}=\delta_{ij}$ and the connection coefficients ${\overline{\Gamma}}_{ik}^j$ are all zero. Comparing Eqs. (\ref{eq:gss}) and (\ref{eq:Gpr}) we can see that the connection $\nabla^{\mbox{PR}}$ is metric semi-symmetric, with
\be \tau_i=-\pi_i/\om,\label{eq:wi}.\ee
We then get
\be \pi_{ij}=\pi_{ji}=\frac{1}{\om^2}\left((\Sigma^2)_{ij}+\frac12 \delta_{ij}\right),\ee
where
\be \Sigma_{ij}=-\epsilon_{ijk}\pi^k,\ee
and the curvature tensor simplifies to
\be R^{\mbox{PR}}_{ijkl}=\frac{1}{\om^2}\Sigma_{ij}\Sigma_{kl}.\label{eq:prc}\ee
From Eq. (\ref{eq:gss}) we get for the torsion
\be {T^{\mbox{PR}}}_{ij}^k=\delta_i^k\tau_j-\delta_j^k\tau_i,\label{eq:tpr}\ee
where $\tau_i$ is given by (\ref{eq:wi}).
\subsubsection{The difference between the teleparallel and the Pryce connection}
The difference of two connections is a tensor. In our case it is a matter of a straightforward calculations to find this tensor for the teleparallel connection $\Gamma$ defined by Eq. (\ref{eq:gg}) and the Pryce connection given by Eq. (\ref{eq:gpr}). We obtain
\be \Gamma_{ij}^k-{\Gamma^{\mbox{PR}}}_{ij}^k=b_i{\Sigma^k}_j,\label{eq:df1}\ee
where
\be {\bf b}=\frac{\bpi\times{\bf \n}}{\om-\bp\cdot{\bf \n}},\quad \n=(0,0,1).\label{eq:df2}\ee
\subsubsection{The Pryce operator}
We now use Eqs. (\ref{eq:qid}) and (\ref{eq:cd}). Denoting
\be {\bf X}=i\frac{\partial}{\partial \bp},\label{eq:xi}\ee
the position operator corresponding to $\nabla^{\mbox{PR}}$ is
\be {\bf Q}^{\mbox{PR}}\doteq i\boldsymbol{\nabla}^{\mbox{PR}}=P {\bf X} P+(I-P){\bf X}(I-P)\label{eq:prp},\ee
or, explicitly,
\be ({\bf Q}^{\mbox{PR}}_j\f)^k=i\partial_j f^k+ \frac{i}{|\bp|}\left(\pi^k\delta_{jl}-\pi_l\delta_j^k\right)f^l.\ee
Using Eqs. (\ref{eq:s}),(\ref{eq:munu}) it can be also written as
\be {\mathbf Q}^{\mbox{PR}}={\mathbf X}+\frac{(\bp\times {\bf s})}{|\bp|^2},\ee
and, in this form, it is known as the Pryce position operator. The following alternative formula for this operator, expressing it in terms of the Poincar\'e group generators, is easily verified and is well known
\be {\bf Q}^{\mbox{PR}}=\frac12\left(\frac{1}{P^0}{\bf N}+{\bf N}\frac{1}{P^0}\right)\label{eq:qn},\ee
where $N^i$ are the boost generators (\ref{eq:N0}).
It should be mentioned that even though the bundle $TV_0^+$ is trivial, the non-triviality of its splitting defined by the helicity-related  projection $P$ is reflected in the fact that the adapted connection has a non-zero curvature, and, as a consequence, the components of the Pryce operator do not commute.
\begin{Remark}
The flat covariant derivative $d$ does not mix the two photon $\pm 1$ helicities. As a consequence the Pryce operator commutes with the helicity operator $\Lambda,$ and not only with its square $\Lambda^2=I-P.$
\end{Remark}
Using now Eq. (\ref{eq:prc}) we get
\be ([Q_i,Q_j]\f)^k=-\frac{1}{\om^2}\Sigma_{ij}{\Sigma^k}_l f^l,\ee
which is usually written as (see e.g. \cite[Eq. (5,9), p. 40]{bacry})\footnote{In Ref. \cite{bacry} helicity is defined with the opposite sign.}
\be [Q_i,Q_j]=i\epsilon_{ijk}\frac{p_k}{\om^3}\Lambda,\ee
where $\Lambda=i\Sigma$ is the helicity operator.
\subsection{The Hawton-Baylis operator}
The teleparallel connection described in Sec. \ref{sec:tpc} leads to the Hawton-Baylis photon position operator $Q^{\mbox{HB}}_i$ with commuting components
\be Q^{\mbox{HB}}_i=i\nabla_i,\label{eq:qnhb}\ee
with the defining relations
\be \nabla_ie_\alpha=0,\,(i=1,2,3),\,(\alpha=1,2,3),\label{eq:qe}\ee
where $e_\alpha$ are  given by Eqs. (\ref{eq:e12})-(\ref{eq:e3}). It follows from Eqs. (\ref{eq:qnhb}) and (\ref{eq:qe}) that the states $e_\alpha(\bp)$ are localized at $\bx=0.$ To obtain states localized at any point $\bx={\bf a}$ we multiply these states by $\exp (-i\bp\cdot{\bf a}).$

From Eqs. (\ref{eq:df1}),(\ref{eq:df2}) we then get\footnote{C.f. \cite[Eq. 53]{bh} and \cite[Eq. 2.49]{dob1}).}
\be Q^{\mbox{HB}}_i\f=Q^{\mbox{PR}}_i\f+b_i\Lambda\f.\ee
In particular
\be Q^{\mbox{HB}}_3=Q^{\mbox{PR}}_3.\label{eq:qrq}\ee
Using this last equality, together with Eqs. (\ref{eq:qn}) and (\ref{eq:N0}) we recover the fact, mentioned after Eq. (\ref{eq:N0}), that the states $\tilde{e}_\alpha=\om^{-1/2}e_\alpha$ are invariant under boosts in the direction of the third axis. Indeed, from the Leibniz property (\ref{eq:lei}) and from (\ref{eq:qe}) we get
\be  Q^{\mbox{HB}}_i\tilde{e}_\alpha=-\frac{i\pi_i}{2\om}\tilde{e}_\alpha.\ee
Thus also \be Q^{\mbox{PR}}_3\tilde{e}_\alpha=-\frac{i\pi_3}{2\om}\tilde{e}_\alpha.\label{eq:q3e}\ee On the other hand, using (\ref{eq:qn}) and the Lie algebra commutation relation $[N^3,P^0]=iP^3$, $Q^{\mbox{PR}}_3$ can be written as
\be Q^{\mbox{PR}}_3=\frac{1}{\om}N^3-\frac{i\pi_3}{2\om},\ee
which, together with (\ref{eq:q3e}),  leads to $(1/\om)N^3\tilde{e}_\alpha=0,$ and thus $N^3\tilde{e}^\alpha=0.$ But this reasoning does not give us the clue as how the basis $e_\alpha$ can be obtained by taking the speed of light limit of the simple polarization basis (\ref{eq:em2}).\footnote{A brief discussion of an application of the Hawton-Baylis photon position operator to optical beams (derived via Wigner's little group method) can be found in Ref. \cite{hawton19}, c.f. also references therein.}

It is straightforward to verify that the operators $Q_i=Q^{\mbox{HB}}_i$ have the axial symmetry:
\be [M^3,Q_1]=iQ_2,\, [M^3,Q_2]=-iQ_1,\, [M^3,Q_3]=0.\label{eq:m3q}\ee
\subsubsection{Photon states localized on circles}
Since the three components of $Q^{\mbox{HB}}$ commute, they can be simultaneously diagonalized, and the simultaneous eigenvalue equation $Q^{\mbox{HB}}_i \f=q_i \f$ has three independent  solutions
\be \f_{\alpha,{\bf q}}(\bp)=\exp(-i{\bf q}\cdot\bp)e_\alpha(\bp),\, ({\mathbf q}\in\BR^3, \alpha=1,2,3).\ee
The states $\f_{1,{\bf q}}\pm i\f_{2,{\bf q}}$ describe photons localized at ${\bf q}\in\BR^3,$  of helicity $\pm 1,$ and $\f_{3,{\bf q}}$ is of helicity $0.$
However, since the states $e_\alpha(\bp)$ have a rather complicated $\bp$-dependence, and owing to the axial symmetry of $Q^{\mbox{HB}},$ it is natural to look for {\bf simple axially symmetric} simultaneous eigenvalue equations for $Q=Q^{\mbox{HB}}$:
\begin{eqnarray}
M^3\, \f&=&0,\label{eq:qc1}\\
Q_3\,\f&=&0,\label{eq:qc2}\\
(Q_1^2+Q_2^2)\,\f&=& R^2\,\f,\,R>0\label{eq:qc3}.
\end{eqnarray}
These would give us states localized on the circle $x^2+y^2=R^2,\, z=0,$ in the Cartesian coordinates $(x,y,z)$ in the photon's position space. To this end it is convenient to introduce cylindrical coordinates $(\rho,\phi,p_3),$ $\rho=\sqrt(p_1^2+p_2^2).$ One can then verify that states of the form
\be \f(\rho,\phi,p_z)=F(\rho)\begin{pmatrix}-\sin(\phi)\\ \cos(\phi)\\0\end{pmatrix}\label{eq:fj}\ee
are simple solutions of Eqs. (\ref{eq:qc1}) and (\ref{eq:qc2}), while Eq. (\ref{eq:qc3}) imposes a second order ordinary differential equation on $F(\rho)$:
\be \rho^2 F''(\rho)+\rho F'(\rho)+(R^2\rho^2-1)F(\rho)=0.\label{eq:Fr}\ee
with the following solution finite at the origin \cite{shi}:
\be F(\rho)=c J_1(R\rho),\ee
where $J_1$ is the Bessel function of the first kind, and $c$ is a constant - see Fig. \ref{fig:2}.
 \begin{figure}[!htb]
  \center
 \includegraphics[width=10.5cm, keepaspectratio=true]{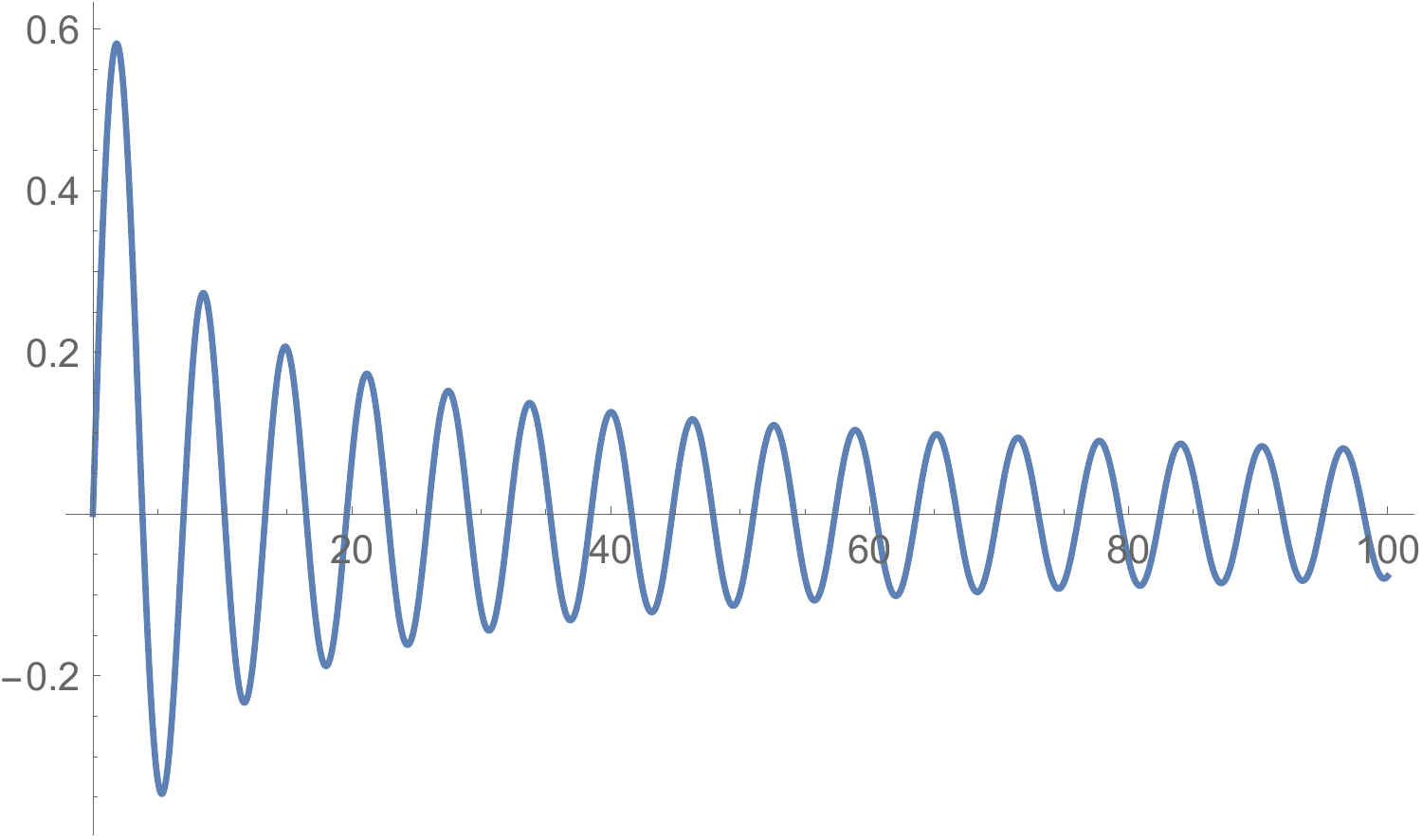}
 \caption{Bessel function $J_1(\rho).$}
\label{fig:2}\end{figure}
For reasons that will be clear from the next paragraph we choose for the constant $c$ the value:
\be c=- i R.\ee
\paragraph{Loop states}
Let $\ell$ be a closed loop in $\BR_\bx^3$ defined by a function $\Br(t),\quad 0\leq t\leq 2\pi.$  Following Ref. \cite{jj} let us define the state $\f_\ell$ by
 \be \f_\ell(\bp)=\frac{1}{2\pi}\int_0^{2\pi}e^{-i\bp\cdot\Br(t)}\,d\Br(t).\label{eq:o}
\ee
 Then $\f_l(\bp)$ is a superposition of simultaneous eigenstates of $X_i$ - see Eq. (\ref{eq:xi}) - ${\bf X}$-localized in $\BR_\bx^3$ at the points of the loop $\ell$. Therefore it is ${\bf X}$-localized on $\ell.$ The fact that the loop is closed, i.e. $\Br(2\pi)=\Br(0),$ implies that the state $\f_\ell(\bp)$ is an (improper) element of $\H_{ph},$ i.e that \be \bp\cdot \f_\ell(\bp)=0.\ee
  Taking the Fourier transform $\f_\ell(\bx)$ of $\f_\ell(\bp)$ we get
 \be \f_\ell(\bx)=\frac{1}{(2\pi)^3}\int_\BR^3 e^{i\bp\bx}\f_\ell(\bp)d^3p=\frac{1}{2\pi}\int_0^{2\pi}\delta\left(\bx-\Br(t)\right)d\Br(t),\label{eq:cc}\ee
 and it is clear $\f_\ell(\bx)$ has its support on the loop in the position coordinates space - it vanishes at all points $\bx$ outside the loop. It is also clear that for two non-intersecting loops $\ell,\ell'$ the states $\f_\ell$ and $\f_{\ell'}$ are orthogonal to each other.\footnote{Loops may form topologically inequivalent knots. In this respect the loop states discussed above are similar to knotted solutions of Maxwell equations discussed in Ref. \cite[Sec. 7]{bb1}.}

 As an example let us take for $\ell$ the circle in $(x,y)$ plane of radius $R$ given by the parametric equations
\be x=R\cos t,\,y=R\sin t,\,z=0.\ee
A straightforward calculation in cylindrical coordinates $(\rho, \phi,p_4)$ (in the momentum space) leads then to
\be \f_{x^2+y^2=R^2}(\bp)=-i\frac{R}{\rho}\, J_1(R\rho )\begin{pmatrix}-p_2\\ p_1\\0\end{pmatrix}.\ee
where $J_1$ is the Bessel function and $\rho=\sqrt{(p^1)^2+(p^2)^2},$ which coincides with the state $\f$ given by the formula (\ref{eq:fj}).
\paragraph{Amrein's washer photon states}
We use Eq. (\ref{eq:cc}) for a circle of radius $R$ at $z=z_0$, thus
\be \begin{split}
x(t)&=R\cos t,\\
y(t)&=R \sin t,\\
z&=z_0.
\end{split}
\ee
Thus \be \f_{R,z_0}=\frac{1}{2\pi}\int_0^{2\pi}\delta(x-R \cos t ,y-R\sin t,z-z_0)\begin{pmatrix}-\sin t\\ \cos t\\0\end{pmatrix} dt.\ee
We will use cylindrical coordinates $(r,\phi,z)$ in the position space. In these coordinates
\be \delta(x-x',y-y',z-z')=\frac{1}{r}\delta(r-r')\delta(\phi-\phi')\delta(z-z'),\ee
while \be d^3x=rdr\,d\phi\,dz.\ee
Thus
\be
\f_{R,z_0}(r,\phi,z)=\frac{1}{2\pi}\int_0^{2\pi}\frac{\delta(r-R)}{r}\delta(z-z_0)\delta(\phi-t)\begin{pmatrix}-\sin t\\ \cos t\\0\end{pmatrix}dt.\ee
After integrating over $dt$ we get
\be \f_{R,z_0}(r,\phi,z)=\frac{1}{2\pi}\frac{\delta(r-R)}{r}\delta(z-z_0)\begin{pmatrix}-\sin \phi\\ \cos \phi\\0\end{pmatrix}.\label{eq:ws}\ee
We now take a continuous superposition of these states for $R$ varying between $R_1$ and $R_2>R_1$ end $z_0$ varying between $z_1$  and $z_2>z_1$. As a superposition of photon states, it will be still a photon state. Let $\chi_{R_1,R_2}(r)$ be the function equal $1$ for $R_1\leq r\leq R_2$ and zero otherwise, and, similarly, let $\chi_{z_1,z_2}$ be the function equal $1$ for $z_1\leq z\leq z_2$ and zero otherwise. Taking the integral
\be \int_{-\infty}^\infty dz\int_0^{\infty}\chi_{R_1,R_2}(R)\chi_{z_1,z_2}(z)\f_{R,z_0}(r,\phi,z)\,dR\,dz\ee
we obtain
\be \f_{R_1,R_2,z_1,z_2}(r,\phi, z)=\frac{1}{2\pi r}\chi_{R_1,R_2}(r)\chi_{z_1,z_2}(z))\begin{pmatrix}-\sin \phi\\ \cos \phi\\0\end{pmatrix}.\ee
The function is evidently square integrable with respect to $d^3x=rdr\,d\phi\,dz$ and its probability density is zero everywhere except for the bolt washer-like region $R_1\leq r\leq R_2$, $z_1\leq z\leq z_2$ - see Fig. \ref{fig:3}.\\
 \begin{figure}[!htb]
  \center
 \includegraphics[width=0.3\textwidth, keepaspectratio=true]{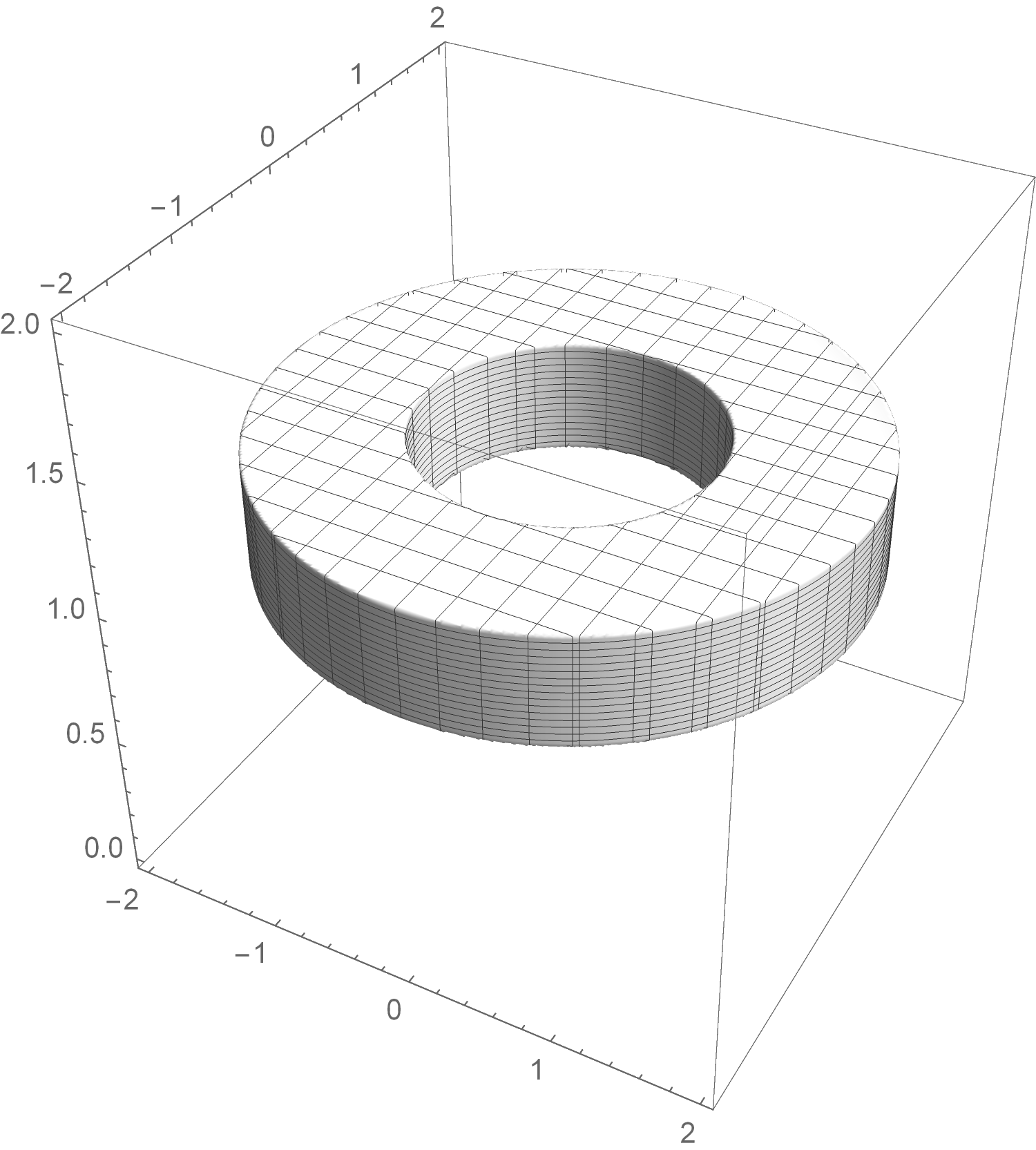}
 \caption{Amrein's washer state (\ref{eq:ws}) is strictly (weakly) localized in the region $1\, \mbox{nm}\leq R\leq \sqrt{2}\, \mbox{nm},\,1\,\mbox{nm}\leq z\leq 1.4\, \mbox{nm}.$}
\label{fig:3}\end{figure}
We call it Amrein's state, as the existence of such states was first proved in 1968 by A.O. Amrein \cite{amrein}. These states are also strictly localized with respect to the commuting position operators $Q_i^{\mbox{HB}}.$
\subsection{POV measure photon's localization}\label{sec:povm}
Every self-adjoint operator admits a spectral decomposition. Usually we write it as:
\be A=\int \lambda dE(\lambda).\ee
More generally, given a family of commuting observables, we have a spectral measure on the common spectrum of these observables. Here we have operators $X_i=i\partial/\partial p^i$ defined on  $\H,$ with commuting components $[X_i,X_j]=0,$ and we have a unique spectral measure on $\BR^3$ such that
\be X_i=\int_{\BR^3}x_i dE(\bx).\ee
Then for every Borel set $\Delta\subset \BR^3$ the operator
\be E(\Delta)=\int_\Delta dE(\bx)\ee
is a projection operator on the subspace of states localized in  $\Delta,$ for any two sets $\Delta$ i $\Delta'$  the operators $E(\Delta)$ i $E(\Delta'$ commute.

However the operators $X_i$ do not leave the subspace $\H_{ph}$ invariant, therefore we have introduced
$Q_i^{\mbox{PR}}$
\be Q_i^{\mbox{PR}}=P X_i P+(I-P) X_i (I-P),\label{eq:qi}\ee
where $(I-P)$ is the orthogonal projection operator on $\H_{ph}.$
Thus we have
\be Q_i^{\mbox{PR}}=\int_{\BR^3}x_i dF(\bx),\ee
where
\be F(\Delta)=P E(\Delta)P+(I-P) E(\Delta)(I-P).\label{eq:fd}\ee
Now $F(\Delta)$ are not any longer projection operators, and for different $\Delta$ they do not commute. Nevertheless they are non-negative operators  $0\leq F(\Delta)\leq 1$ and
\be \int_{\BR^3} dF(\bx) = I,\ee
Therefore we have a  POV - positive operator valued measure. Jauch and Piron \cite{jp} called a photon state $\f$ weakly localized in $\Delta$ if $F(\Delta) \f=\f,$ and conjectured existence of such states. Amrein \cite{amrein} provided a rigorous proof of their existence and has shown how to construct them. Our formula (\ref{eq:ws}) provides a rich explicit family of such states.
Notice however that the circle states $\f_{x^2+y^2=R^2}(\bp)$ while satisfying
\be (Q^{\mbox{HB}}_1)^2+(Q^{\mbox{HB}}_2)^2)\,\f_{x^2+y^2=R^2}= R^2 \f_{x^2+y^2=R^2},\ee
they are not eigenstates of $(Q^{\mbox{PR}}_1)^2+(Q^{\mbox{PR}}_2)^2.$ Instead, for any real measurable function $\phi$ on $\BR^+$, they are eigenstates to the eigenvalue $\phi(R^2)$ of the operator $Q_{\phi(x^2+y^2)}$ defined as
\be\begin{split}  Q_{\phi(x^2+y^2)}&\doteq P\phi((X_1)^2+(X_2)^2)P+(I-P)\phi((X_1)^2+(X_2)^2)(I-P)\\&=\int \phi(x^2+y^2)\, dF(\bx),\end{split}\ee
where $dF(\bx)$ is the POV measure defined by Eq. (\ref{eq:fd}).\footnote{For instance in Ref \cite{koczan} the square root function is proposed.}

The circle localized states $\f_{x^2+y^2=R^2}$ are superpositions of helicity $+1$ and helicity $-1$ states, in agreement with Theorem 2 of Ref. \cite{amrein}. Their projections on definite helicity subspaces, $+1$ or $-1$, are still circle-localized with respect to the Hawton-Baylis operator, but they are not weakly localized in the sense of the POV measure $F(\Delta)$ - in agreement with Theorem 1 of Ref. \cite{amrein}
\section{Conclusions}
In conclusion, using differential geometric structures on the mass hyperboloid and on the light cone in momentum space we derived the explicit form of the unitary representation of the Poincar\'e group for helicity zero and helicity $\pm 1$ massless particles. Using this explicit form for the boost in the direction of the third axis we have found that simple photon polarization states based on Hertz-type potentials survive the light speed limit and generate a polarization basis used in the construction of photon position operators with commuting components. We have compared these operators, as well as the underlying affine connections, to the classical Pryce operator and connection, and have found that the Pryce connection, Eq. (\ref{eq:pdp}), non-flat, but with rotational symmetry, is metric semi-symmetric, while the Hawton-Baylis flat, but only axially symmetric, connection, Eq. (\ref{eq:gg}), does not have this property. We have constructed finite-norm photon states localized in bolt washer-like regions and proved that they are strictly localized in these regions with respect to Hawton-Baylis position operators and also with respect to the Jauch--Piron--Amrein POV measures. They are also $z$-localized (but not radius-localized) with respect to the Pryce photon position operator.
\section*{Acknowledgements}
Thanks are due to M. Schlichtinger and J. Szulga for the discussion of the preliminary version of the part of Sec. \ref{boosts}, as well as to G. Koczan, P. Lescaudron and R. Coquereaux for their interest and for a stimulating discussion on related subjects. I wish to thank my wife, Laura, for reading the manuscript and for her constant support.\\
\noindent This work was supported by Quantum Future Group Inc.

\end{document}